\documentclass[doublecol]{epl2}
\usepackage{graphicx} 
\usepackage{xcolor} 
\usepackage{mathtools}
\usepackage{amssymb, amsmath, amsfonts,dsfont}
\usepackage{physics}
\usepackage{bm} 
\usepackage{subcaption}
\usepackage{comment}
\usepackage{float}
\usepackage{cleveref}
\usepackage{soul}
\usepackage{cuted}
\usepackage[normalem]{ulem}
\usepackage{amsthm}

\newcommand{\red}[1]{\textcolor{red}{#1}}

\makeatletter
\newcommand{\vast}{\bBigg@{3}}
\newcommand{\Vast}{\bBigg@{4}}
\newcommand{\Vastt}{\bBigg@{5}}
\makeatother

\newtheorem*{Thm*}{{\it\bfseries Main result}}

\newtheorem{Lem}{Lemma}

\newtheorem{Theorem}{Theorem}

\newcommand\restr[2]{{
  \left.\kern-\nulldelimiterspace 
  #1 
  \right|_{#2} 
  }}
\makeatletter

\newcommand{\maintitle}{Spectral criteria for generalization in unsupervised Hebbian nets}
\title{\maintitle}

\title{\maintitle}
\shorttitle{\maintitle}
\author{Elena Agliari\inst{1} \and Paulo Duarte Mourão\inst{1}, Alberto Fachechi\inst{1} \and Pierpaolo Vivo\inst{2}}
\shortauthor{Elena Agliari \etal}
\institute{                    
  \inst{1} Dipartimento di Matematica, Sapienza Universit\`a di Roma, Rome, Italy.\\
  \inst{2}Department of Mathematics, King’s College London, The Strand, London WC2R 2LS, UK.
}

\abstract{ We consider an unsupervised Hebbian network where the pairwise interactions among neurons are built on noisy realizations of hidden ground-truth vectors. Unlike classical Hopfield models, designed as memory devices, this class of networks can be employed to extract latent structure and generalize beyond the ``training'' set. By combining random matrix theory and replica methods, we derive the asymptotic spectrum of the corresponding interaction matrix and show that the onset of generalization is controlled by a sharp spectral transition. Depending on the quality and the size of the accessible dataset, the spectrum displays either two separated bulks, encoding informative and noisy directions, or a merged single-bulk phase where such distinction is lost. We show that, when coupled with regularization, the emergence of such a spectral split predicts the network’s capability to reconstruct the ground-truth vectors from corrupted samples.
}

\begin{document}
\setstcolor{red}
\maketitle




\section{Introduction}
 Associative neural networks provide a paradigmatic framework to investigate how collective systems can store, retrieve, and generalize information. 
In particular, the Hopfield network, being inspired by spin-glass models, describes a set of binary neurons interacting pairwise to minimize the system's overall energy; the latter is specified by a coupling matrix, trained upon a set of data patterns, in such a way that, when a certain query is given as input, the neurons iteratively rearrange to reach a stable state which corresponds to the output associated to the input. In the simplest scenario, the task is the reconstruction of a set of patterns that are available and directly stored in the interaction matrix, typically by Hebb's rule. More challenging tasks include sequence retrieval, categorization, generation, disentanglement, and much more, see e.g., \cite{Branchtein-JPA1992, Leuzzi_2022,AAKBA-EPL2023,Agliari_HeteroTAM,Negri-Physa2025}.  In standard analytical formulations, patterns are assumed to be (high-dimensional) orthogonal, i.i.d., random vectors. While ensuring tractability, this assumption limits the validity of the results in realistic scenarios, where correlations, redundancies, or latent organization displayed by empirical data are essential features for information processing. 
\par
Recently, increasing attention has been devoted to structured datasets, in order to understand how their internal organization shapes the emerging properties of associative neural networks and triggers concept formation and feature learning, shifting from pure memorization to inference capabilities \cite{AABD-NN2022, Negri-PRL2023,AAAF-NN2024,benedetti2026}. 
Despite these novel objectives, the mathematical framework remains essentially unchanged: the information encoded in the network weights can still be regarded as random samples drawn from a probability distribution in a high-dimensional space. However, the presence of underlying correlations violates the independence assumptions on which most analytical results rely. Extending such results beyond the i.i.d.\ setting therefore constitutes a fundamental challenge in the rigorous analysis of neural networks trained on empirical data.
\par
Moving in this direction, we consider a controlled yet non-trivial dataset made of noisy realizations (examples) of unknown ground-truth patterns (archetypes), shifting the task from standard pattern retrieval to \emph{generalization}: the objective is no longer the recovery of a specific, stored pattern, but the reconstruction of the underlying ground-truth from its corrupted realizations. Our analysis starts by deriving the asymptotic spectral properties of the Hebbian interaction matrices built on the noisy samples. 
We find that generalization in these networks is triggered by a \emph{spectral transition}: depending on the quality and size of the sample, the asymptotic spectrum exhibits either two separated components -- corresponding to, respectively, informative and non-informative eigenvectors -- or a merged single-bulk -- where such distinction is lost, hampering generalization. We support this picture through a combination of analytical calculations and Monte Carlo (MC) simulations of the neural dynamics. 

\section{The unsupervised Hebbian rule}\label{sec:Definitions}
The Hopfield model is a network of binary spins, whose state we denote by $\bm \sigma = (\sigma_1,\ldots,\sigma_N)\in\{-1,+1\}^N$, for some $N\in\mathbb N$, updated as
\begin{equation} \label{eq:neurdyn}
    \bm\sigma^{(n+1)}=\text{sgn}\big(\bm J\cdot \bm\sigma^{(n)}\big).
\end{equation}
The interaction matrix $\bm J$ is designed to make the network able to perform retrieval tasks. More precisely, given $K$ binary vectors $\bm\xi^\mu$, with $\mu=1,\ldots,K$, interpreted as memories, we say that the network retrieves the pattern $\bm\xi^\mu$ 
if, by initializing the neuronal configuration $\bm{\sigma}^{(0)}$ ``close'' to the target pattern, the neuronal dynamics in Eq.~\eqref{eq:neurdyn} eventually converges to a fixed point reconstructing the stored pattern, i.e. $\bm{\sigma}^{(\infty)} = \bm\xi^\mu$. 
In analytical frameworks, memory entries are commonly assumed to be independently drawn from a Rademacher distribution
\begin{equation}\label{eq:stat_xi}
    P\left(\xi^\mu_i=\pm1\right)=\frac12.
\end{equation}
In this setting, a convenient choice for the coupling matrix is given by Hebb's rule
\begin{equation} \label{eq:J_H}
J^H_{ij}\doteq\frac1N\sum_{\mu=1}^K\xi^\mu_i\xi^\mu_j,\quad  J^H_{ii}=0.
\end{equation}
This ensures that each single stored memory is a fixed point as long as the initial configuration is not too far from the target pattern and the network load, defined as
\begin{equation}\label{eq:alpha}
    \alpha\doteq \lim_{N\to\infty}\frac KN,
\end{equation}
is sufficiently small.
In order to mitigate these limitations, several modifications of the Hebbian prescription have been proposed. Among them the following\footnote{This interaction matrix was originally introduced to mimic the interplay between awake and resting regimes, where the network is, respectively, exposed to new patterns and subjected to removal and consolidation mechanisms; following this inspiration, $t$ is also interpreted as ``sleeping time'' \cite{FAB-NN2019}.}
\begin{equation} \label{eq:J_D}
     J^D_{ij}\doteq\frac{1}{N}\sum_{\mu,\nu=1}^{K}\xi^\mu_i 
    \left(\frac{1+t}{\bm I+t \bm C}\right)_{\mu\nu}
    \xi^\nu_j,
    \quad J^D_{ii}=0, 
\end{equation}
with $t \in \mathbb R^{+}_0$ and $C_{\mu\nu}=\sum_i \xi_i^{\mu}\xi_i^{\nu}/N$, interpolates between Hebb's ($t=0$) and Kohonen's ($t \to \infty$) rules.  
Increasing $t$ in the definition of $\boldsymbol J^D$ reduces the stability of spurious configurations, namely mixtures of stored patterns whose retrieval is regarded as an error of the machine \cite{FAB-NN2019,agliari2024spectral}. 
\par
Moving from a setting where the designer has full access to the patterns $\{ \bm \xi^{\mu}\}_{\mu=1,...,K}$ to more realistic situations, we assume that the only available information consists of noisy realizations of the original patterns denoted by $\{\tilde{\bm{\xi}}^\mu_a\}^{\mu=1,...,K}_{a=1,...,M}$. Again, retaining a synthetic and controllable setting, we generate these accessible vectors as
\begin{equation}\label{eq:def_eta}
    \tilde{\xi}^\mu_{a,i}=\xi^\mu_i\chi^\mu_{a,i},
\end{equation}
with $P\left(\chi^\mu_{a,i}=\pm1\right)=\tfrac12({1\pm r})$, 
where the parameter $r\in(0,1]$ tunes the \emph{quality} of the sample. Hereafter, we will refer to the ground-truth patterns as \emph{archetypes} and to the noisy-realizations as \emph{examples}. Within this setting, $\operatorname{Cov}\big(\tilde\xi^\mu_{a,i}, \tilde \xi^\nu_{b,j}\big) =
\delta_{ij}\delta_{\mu\nu}
\left[\delta_{ab}+ r^2 (1-\delta_{ab})\right].$
Having in mind an unsupervised scenario, the label $\mu$ in these examples is latent \cite{AABD-NN2022} and it is thus convenient to relabel the examples by a multi-index $\ell = (\mu,a) \in \{1, ..., MK \}$. In this framework, the coupling matrices \eqref{eq:J_H} and \eqref{eq:J_D} are extended, respectively, as 
\begin{equation} \label{eq:J_H_r}
    \tilde{J}^H_{ij}=\frac{1}{N M}\sum_{\ell=1}^{M K} \tilde \xi^{\ell}_i 
    \tilde\xi^{\ell}_{j},
    \qquad \tilde{J}^H_{ii}=0,
\end{equation}
and
\begin{equation} \label{eq:J_D_r}
    \tilde{J}^D_{ij}=\frac{1}{N M}\sum_{\ell, \ell'=1}^{M K} \tilde \xi^{\ell}_i 
    \left(\frac{1+t}{\bm I+t \tilde{ \bm C}}\right)_{\ell \ell'}
    \tilde\xi^{\ell'}_{j},
    \qquad \tilde{J}^D_{ii}=0,
\end{equation}
with $\tilde C_{\ell \ell'}=\tfrac1 {NM}\sum_{i=1}^N \tilde \xi^\ell_{i}\tilde\xi^{\ell'}_{i}$ being the {\it example} correlation matrix. Here, the task consists in the reconstruction of the archetypes while the retrieval of a specific training example is interpreted as overfitting \cite{AAAF-NN2024}. Indeed, the coupling matrix \eqref{eq:J_D_r} can also be recovered as a minimizer of a squared-error loss, where $t$ acts as (the inverse of) a regularization parameter, so that large values of $t$ correspond to weak regularization, which in turn may expose the system to the risk of overfitting.

\begin{figure*}[t]
\centering
  \centering
  \includegraphics[width=.30\linewidth]{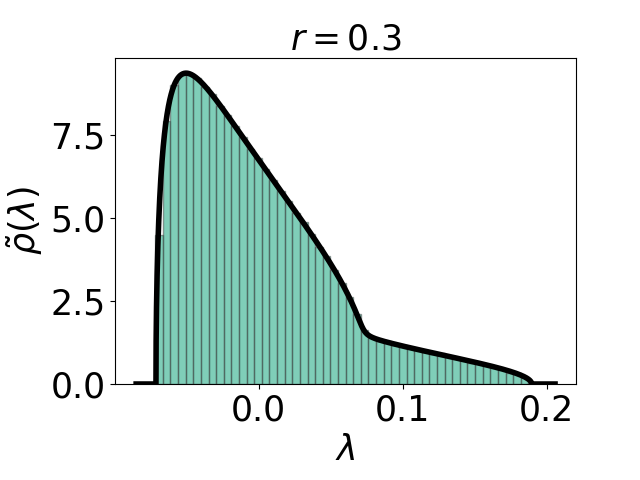}
\hspace{0cm}
  \centering
  \includegraphics[width=.30\linewidth]{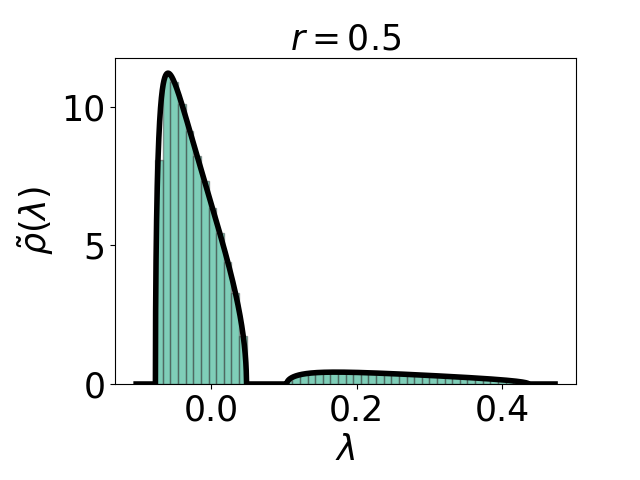}
\caption{Spectrum of the unsupervised model characterized by the interaction matrix $\boldsymbol{\tilde J}^H$ in Eq. \eqref{eq:J_H_r}, for $r=0.3$ (left) and $r=0.5$ (right). The histogram in green represents the data resulting from computing the eigenvalues of interaction matrices with $N=1000$ neurons across $50$ independent samples, while the solid black line shows the corresponding theoretical prediction. The load and number of examples were fixed at $\alpha=0.1$ and $M=50$, respectively.}
\label{fig:spectra}
\end{figure*}

\section{Spectral theory of the unsupervised Hebbian matrices}\label{sec:spec}
While a comprehensive statistical-mechanics picture of the emergent behavior is available for the models defined by \eqref{eq:J_H} and \eqref{eq:J_D}, an analogous theoretical understanding is still lacking for the models corresponding to \eqref{eq:J_H_r} and \eqref{eq:J_D_r}. In what follows, we derive the asymptotic spectral distribution for the regularized, unsupervised, Hebbian matrix $\tilde{\bm{J}}^{D}$ and infer from it how the quality $r$, the number of examples per class $M$ and the regularizer $t$ affect the performance of the network.
 
To this goal, it is technically more convenient to start with the unregularized case, where the diagonal constraint is relaxed; we denote it as $\tilde {\bm J}^{H'}$.\footnote{In general, we use primed variables to denote quantities associated with non-zero diagonal versions of the models. On the other hand, tilded variables indicate that there exists noise in the dataset as per Eq. \eqref{eq:def_eta}.}
Then, for a fixed realization of the dataset, we can write its empirical spectral measure as $ \frac 1N \sum_{\alpha} \delta_{\lambda, \tilde\lambda'_\alpha}$, with $\tilde\lambda'_\alpha$ being the generic eigenvalue. 
When $N\to\infty$, the empirical spectral measure is expected to converge in weak-${}^*$ topology to a deterministic limit 
\begin{equation}
    \label{eq:det_lim_rho}
    \tilde\rho'(\lambda)\doteq\lim_{N\to\infty} \mathbb E_{\boldsymbol  \xi}\frac 1N \sum_{\alpha } \delta_{\lambda,\tilde\lambda_\alpha'}.
\end{equation}
We follow the Edwards--Jones formalism \cite{edwards1976eigenvalue}, whereby the asymptotic distribution $\tilde\rho'(\lambda)$ is obtained by computing the quenched free energy of a suitable Gaussian spin-glass model:
$$
\tilde\rho'(\lambda) = -\frac 2\pi \lim_{N\to\infty} \lim_{\epsilon\to 0^+} \operatorname{Im}\frac{\partial}{\partial\lambda_\epsilon}\frac1N \mathbb E_{\boldsymbol  \xi} \log Z_N (\lambda_\epsilon),
$$
with $\lambda_{\epsilon}=\lambda-i\epsilon$, and the partition function being $Z_N (\lambda_\epsilon ) = \int _{\mathbb R^N}d \boldsymbol y \exp\Big(-\frac{i}{2}\boldsymbol y^T (\lambda_\epsilon \bm I-{\tilde{\bm J}^{H'}})\boldsymbol y\Big).$
%
%
%
The computation of the associated quenched free energy is performed by replica-trick, as detailed in the Supplementary Material (SM). 
Here, we directly state the
\begin{Thm*}\label{thm:1}
    Let $\mu_1 = \frac{1-r^2}{M}$ and $\mu_2 =r^2+\mu_1$. For any $\lambda\in \mathbb R$, define the functions
    \begin{align*}
    a(\lambda)&=\lambda\mu_1\mu_2\\
    b(\lambda)&=(\alpha M-1)\mu_1\mu_2-\lambda\left(\mu_1+\mu_2\right)\\
    c(\lambda)&=[1-\alpha(M-1)]\mu_1+(1-\alpha)\mu_2+\lambda,
\end{align*}
and $D(\lambda)= u(\lambda)^2+v(\lambda)^3$ with
\begin{align}
    u(\lambda)&=\frac{2b(\lambda)^3-9a(\lambda)b(\lambda)c(\lambda)-27a(\lambda)^2}{54a(\lambda)^3},\\
    v(\lambda)&=\frac{3a(\lambda)c(\lambda)-b(\lambda)^2}{9a(\lambda)^2}.
\end{align}
Then, the asymptotic spectral distribution 
of the unsupervised Hebbian ensemble \eqref{eq:det_lim_rho}
, in the thermodynamic limit with $\lim_{N\to\infty} K/N=\alpha \ge 0$ and $\alpha M\ge 1$, reads
\small{
\begin{equation*}
    \tilde\rho' (\lambda) = \frac{\sqrt{3}}{2\pi}\big(\sqrt[3]{\sqrt{D(\lambda)}+u(\lambda)}+\sqrt[3]{\sqrt{D(\lambda)}-u(\lambda)}\big)\bm 1_{D(\lambda)>0} .
\end{equation*}
}
\end{Thm*}
This statement provides an explicit expression for the limiting law in the full-rank regime, assuming convergence of the empirical spectral distribution as $N\to\infty$. 
In the low-rank case $\alpha M<1$, a $\delta$-peak at $\lambda =0$ with mass fraction $1-\alpha M$ appears, while non-zero eigenvalues are still described by the explicit expression of $\tilde\rho'(\lambda)$ given above.
Once the limiting law $\tilde \rho'(\lambda)$ is known, we can get its regularized version $\tilde\rho'_t(\lambda)$ by exploiting the bijective relation between the related families of eigenvalues $\{\tilde \lambda'_\alpha\}_\alpha$ and $\{\tilde \lambda'_\alpha(t)\}_\alpha$, reading \cite{agliari2024spectral}
\begin{equation}\label{eq:bij}
\tilde \lambda'_\alpha (t)=f_t(\tilde\lambda'_\alpha)= \frac{1+t}{1+t\tilde \lambda'_\alpha}\tilde\lambda'_\alpha.
\end{equation}
Thus, for $t>0$, the asymptotic spectral distribution is obtained as the pushforward measure of $\tilde \rho'$ by $f_t$, namely $\tilde\rho'_t(\lambda)=\tilde\rho'(f_t^{-1}(\lambda))\frac{df_t^{-1}(\lambda)}{d\lambda}.$ 

\begin{figure*}[t]
    \centering
    \includegraphics[width=0.5\textwidth]{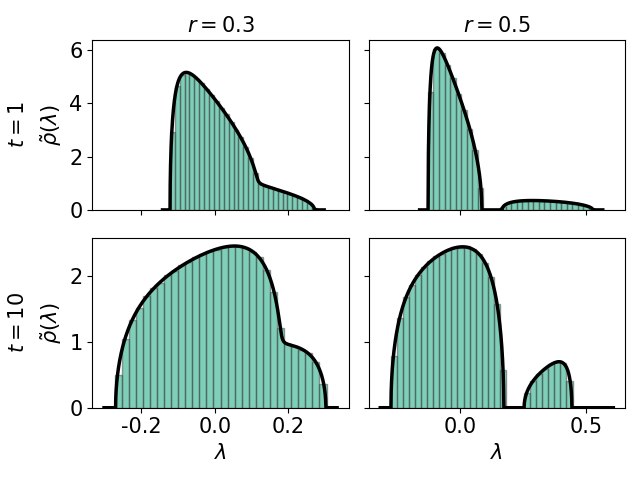}
\caption{Spectrum of the unsupervised regularized model characterized by the interaction matrix $\boldsymbol{\tilde J}^D(t)$ in Eq.~\eqref{eq:J_D_r}, for $r=0.3$ (left) and $r=0.5$ (right), and $t=1$ (above) and $t=10$ (below). The histogram in green represents the data resulting from computing the eigenvalues of interaction matrices with $N=1000$ neurons across $50$ independent samples, while the black solid line shows the corresponding theoretical prediction. The load and number of examples were fixed at $\alpha=0.1$ and $M=50$, respectively.}
\label{fig:spectra_dreaming}
\end{figure*}
%
\par
We now properly handle $\tilde \rho'$ and $\tilde \rho_t'$ to recover $\tilde \rho$ and $\tilde \rho_t$. For the former, we simply shift the asymptotic spectral distribution by a quantity $-\alpha$, since $\tilde{ \bm J}^{H'} = \tilde{ \bm J}^{H}+\alpha \bm I$. The agreement between this theoretical prediction and the numerics is reported in Fig. \ref{fig:spectra} for specific choices of $\alpha$, $M$ and $r$.
\begin{figure*}[t]
    \centering
    \includegraphics[width=0.6\textwidth]{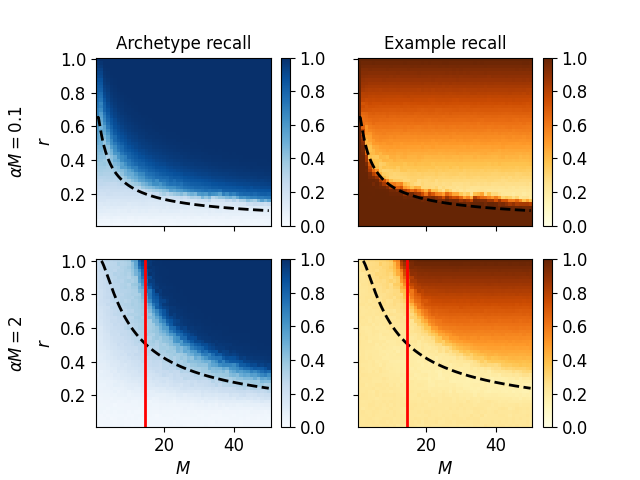}
    \caption{MC simulations for the unsupervised Hopfield model run until convergence for the recovery of archetypes (left) and examples (right) for various values of $r$ and $M$, in the low-rank (above) and full-rank (below) scenarios. The color maps show the archetype overlap ($m_{\bm \xi} = \frac1N \sum_{i=1}^N \sigma^{(\infty)}_i\xi^\mu_i$, left column) and the closest-example overlap ($m_{\tilde{\bm \xi }} = \frac1N \sum_{i=1}^N \sigma^{(\infty)}_i \tilde{\xi}_{a,i}^{\mu}$, right column) computed upon convergence to fixed points $\boldsymbol \sigma^{(\infty)}$, and averaged across $50$ samples. We used $N=1000$ and the starting state had a quality of $p=0.9$ compared to a stored example. The black dashed line marks the separation of the spectrum into two peaks (predicted theoretically by Eq.~\eqref{eq:spectral_gap_crit}), while the vertical red line marks the value $M$ yielding $\alpha = 0.138$, namely the Hopfield model critical load.}
    \label{fig:gen_rm}
\end{figure*}
Remarkably, depending on the combination of the control parameters, $\tilde{\rho}$ exhibits different behaviors (hereafter, unless otherwise specified, we assume the full-rank case $\alpha M\ge 1$). For low $M$ and $r$, 
it consists of a single continuous bulk of non-zero eigenvalues; conversely, for relatively large $M$ and/or $r$, i.e. for sufficiently informative datasets, the distribution consists of two separate bulks, the highest one carrying a fraction $\alpha$ of eigenvalues. In this two-bulk regime, the top empirical eigenspace is expected to have macroscopic overlap with the archetype subspace, while in the merged-bulk regime the archetype directions are not spectrally isolated.
\footnote{In the low-rank setting ($\alpha M<1$), $\textrm{Span}[\{ \tilde{\boldsymbol \xi}^{\ell}\}_{\ell=1,...,MK} ]$ does not cover the full $N$-dimensional space; the remaining subspace is orthogonal to both the patterns {\it and} the examples.} The support of the spectral density gets disconnected for the critical value of the dataset quality $r_c(\alpha,M)$ being the solution of the equation
\begin{equation}\label{eq:spectral_gap_crit}
 \alpha = \frac{(\mu_2-\mu_1)^2}{M \big(\sqrt[3]{(1-\tfrac1M)\mu_1^2}+\sqrt[3]{\tfrac1M \mu_2^2}\big)^3},   
\end{equation}
with $\mu_{1,2}$ as defined in the main result, see also \cite{burda_signalnoise}. 
Indeed, in the extreme case $r=1$ (where examples are identical to the archetypes) the lower bulk collapses to a Dirac $\delta$ at $\lambda=-\alpha$, and the whole distribution reproduces the shifted Marchenko-Pastur distribution at scale factor $\alpha$; in the opposite limit $r=0$, one again recovers the Marchenko-Pastur distribution, but with scale factor $\alpha M$, as we are storing each of the (now uncorrelated) examples independently. 
\par
A similar behavior is observed in the regularized case. Here, the correction to be implemented to $\tilde{\rho}_t'$ is not trivial as the diagonal entries in ${\tilde {\bm J}^{D'}}$ are not constant; yet they do self-average around the first moment $\bar \lambda \doteq \int \lambda {\tilde\rho'_t}(\lambda)d\lambda$, and one can show that the asymptotic spectral distributions of  $\tilde {\bm J}^{D'}$ and $\tilde{\bm J}^{D}+\bar \lambda \bm I$ are the same 
(see the SM for more details). 
Then, the limiting law for $\tilde {\bm J}^D(t)$ in Eq. \eqref{eq:J_D_r} is recovered as $\tilde\rho_t(\lambda)=\tilde\rho'_t(\lambda+\overline{\lambda})$.
%
%
%
Again, the theoretical predictions agree with the numerical estimates of the empirical distribution, see Fig. \ref{fig:spectra_dreaming}.

\section{Application of the theoretical results and numerical checks}\label{sec:MC}
The phenomenology traced so far for the spectral distribution of the unsupervised Hebbian matrices suggests that the separation of two components (which is a genuine dataset-dependent condition) can be leveraged to distinguish between ``signal'' (archetypes) and ``noise'' (any component orthogonal to the archetype span), as identifying the top $K$ eigenvectors of the coupling matrix allows one to recover information about the hidden archetypes. 
Furthermore, we can anticipate how the spectral properties of the coupling matrix impact the outcome of the neuronal dynamics \eqref{eq:neurdyn}: the dot product $\bm J \cdot \bm \sigma^{(n)}$ can be expanded according to the spectral decomposition theorem and, if the spectrum is split in two bulks (this is referred to as ``spectral gap'' from now on), the contributions stemming from the top eigenvalues, that are strongly correlated with archetypes, will prevail over those accounting for the intrinsic noise in the dataset. Thus, by iterating the process, we expect that the gap can have a beneficial effect on archetype recall. Conversely, when the peaks merge, the separation between dominant and subdominant eigenmodes diminishes, in such a way that signal and noise are no longer clearly disentangled; in this regime, the iterative dynamics tends to mix archetype-correlated and noise components, thereby impairing generalization.
\par
In this context, recalling Eq.~\eqref{eq:bij}, we stress that the regularization only induces a continuous deformation of the asymptotic spectral distribution, while its qualitative structure is left unchanged. 
Thus, tuning $t$ is not effective for isolating the components associated with hidden patterns. On the other hand, in the two-bulk regime, regularization can be leveraged to enhance the relative weight of the eigenspace associated with the signal, while suppressing the contribution stemming from the noisy bulk (see the right column of Fig.~\ref{fig:spectra_dreaming}).
%
We also recall that, for $t \to \infty$, the high-rank matrix $\tilde{\bm J}^D$ converges to the null matrix, hence an optimal range of $t$ is expected. 
%
In the following, we provide experimental evidence supporting this picture.

\begin{figure*}[t]
    \centering
    \includegraphics[width=0.7\textwidth]{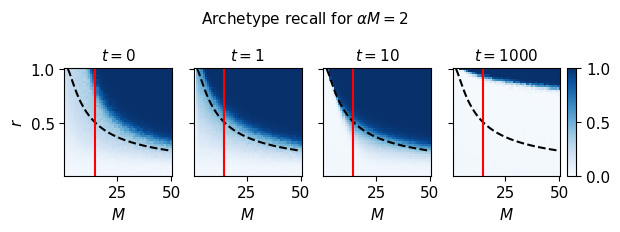}
    \caption{MC simulations run until convergence for the recovery of archetypes for various values of $r$, $M$ and $t$ in the full-rank scenario ($\alpha M=2$). The color maps report the value of the archetype overlap ($m_{\bm \xi} = \frac1N \sum_{i=1}^N \sigma^{(\infty)}_i\xi^\mu_i$) computed upon convergence to fixed points and averaged over $50$ independent samples. We used $N=1000$ and the starting state had a quality of $p=0.9$ compared to a stored example. The black dashed line marks the separation of the spectrum into two peaks (predicted theoretically), while the vertical red line marks the value $M$ for which $\alpha = 0.138$.}
    \label{fig:gen_rm_manyts}
\end{figure*}

%
\par
In Fig. \ref{fig:gen_rm}, we show a comparison between the theoretical predictions coming from spectral analysis and MC simulations for low rank ($\alpha M<1$, first row) and full rank ($\alpha M>1$, second row). 
In the numerical experiments in this section, we prepare the network very close to one of the stored examples\footnote{The initial configuration is generated according to $\sigma^{(0)}_i=\tilde{\xi}^\mu_{a,i}\phi_i$, with $P(\phi_i=\pm1)=(1+p)/2$ and $p=0.9$ for some fixed $\mu$ and $a$. } and run the dynamics \eqref{eq:neurdyn} until convergence;\footnote{Since the temperature is $0$, with parallel updating used here to run \eqref{eq:neurdyn}, convergence here always means convergence to either $1$ or $2$-cycles. In our case, the only $2$-cycles that are observed are oscillations between one state and its symmetric, and only occur for very high $t$, where retrieval is no longer possible.} then we compute the (normalized) overlap with the closest hidden archetype (left column) and the corresponding stored example (right column).
In the low-rank case, the separation of the peaks predicts the emergence of generalization (where the final overlap with the hidden archetype is close to $1$, dark blue regions in the left plots) as opposed to overfitting (where dynamics retrieves the stored examples instead, dark yellow regions in the right plots). In the high-rank regime, generalization is still possible, provided that the load is not too high, a reference value being $\alpha = 0.138$ (highlighted by the vertical lines in the plots), marking the onset of the ``catastrophic forgetting'' in the Hopfield model. As expected, for $r=1$, this is still a sensitive threshold above which reconstruction capabilities are lost, while for noisy datasets ($r<1$) it only provides an upper bound for affordable loads. Furthermore, overfitting becomes impossible to observe, since the number of stored examples per neuron $KM/N=\alpha M$ is also far above the critical load.\par
In the second experiment, we investigate the role of $t$ in enhancing generalization capabilities.
To do this, we focus again on the full-rank case $\alpha M>1$, as it  exhibits a lack of reconstruction performances within the admissible region suggested by spectral analysis.
We thus ran the same experiments as before with the coupling matrix $\tilde {\bm J}^D(t)$ (Eq. \ref{eq:J_D_r}) for different values of $t$, with the results presented in Fig. \ref{fig:gen_rm_manyts}. Remarkably, our results show that tuning $t$ enlarges the reconstruction region and saturates, for some optimal choice $t^*$, the threshold provided by spectral analysis. In this sense, the transition to a double-bulk character of the asymptotic spectral distribution provides a meaningful prediction for reconstruction capabilities of the model: as long as the spectral gap is present (i.e., for sufficiently informative datasets), $t$ can be tuned to preserve generalization even in the presence of a relatively large number of archetypes, a regime typically associated with detrimental interference among the stored patterns and the consequent emergence of a glassy behavior in the neuronal dynamics. This is consistent with the phenomenology observed in the model associated with $\boldsymbol J^D(t)$, where increasing $t$ is found to mitigate these harmful glassy effects \cite{FAB-NN2019,agliari2019dreaming}. 

\section{Designing the optimal coupling: a practical recipe} \label{sec:design}
%

We showed that $t$ can optimize the network performance. However, relating $t^*$
 to the spectral properties of the coupling matrix is generally nontrivial, since the nonlinear neuronal dynamics \eqref{eq:neurdyn} makes the fixed points difficult to control from spectral information alone. Nonetheless, we can try to write down a heuristic expression based on the experience collected so far. For this purpose, we conduct generalization experiments slightly above the threshold yielding the spectral gap:
\begin{equation}\label{eq:sample}
    r = r_c(\alpha,M)+0.05,
\end{equation}
where $r_c(\alpha,M)$ is obtained by {solving Eq. \eqref{eq:spectral_gap_crit} in terms of $r$}.\footnote{It is worth noting that, in the coordinates defined by Eq.~\eqref{eq:sample}, every order parameter changes as $M$ is varied.} The results of numerical simulations for fixed $\alpha M$ are shown in Fig. \ref{fig:gen_dream_levels} (upper row), together with the level sets for the difference between the maxima of the two spectral bulks (bottom row), for $\alpha M = 2$ (left) and $\alpha M=5$ (right), where we can observe a very strong resemblance between both sets of graphs.\footnote{We stress that, increasing $M$ for fixed $\alpha M$ and $r=r_c(\alpha,M)+0.05$ results in both lower quality and lower load.}
\begin{figure*}[t]
    \centering
        \centering
        \includegraphics[width=0.30\textwidth]{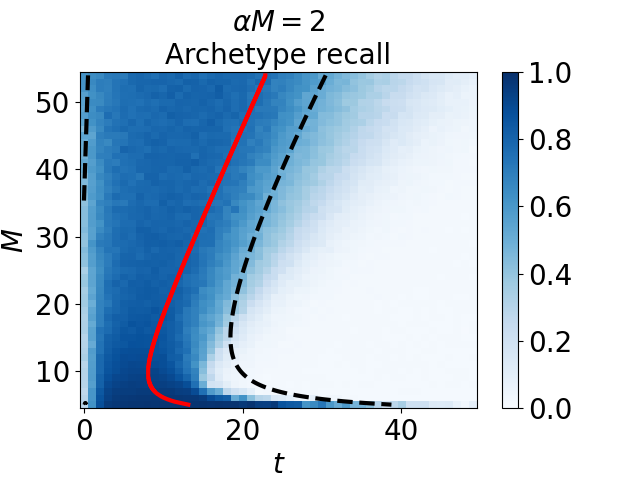}
    \hspace{0cm}
        \centering
        \includegraphics[width=0.30\textwidth]{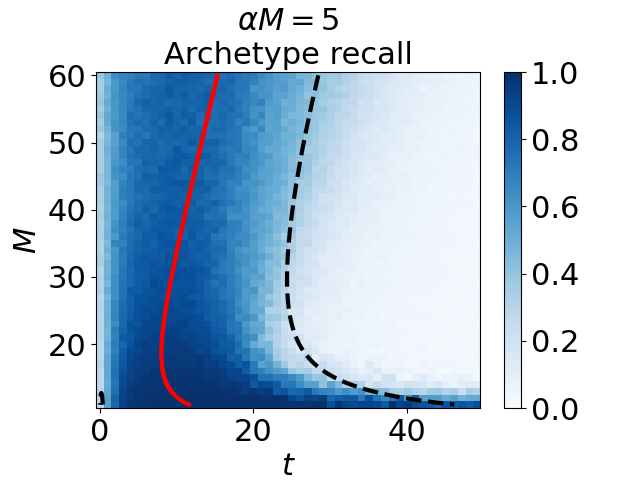}\\
    \vspace{0cm}
        \centering
        \includegraphics[width=0.30\textwidth]{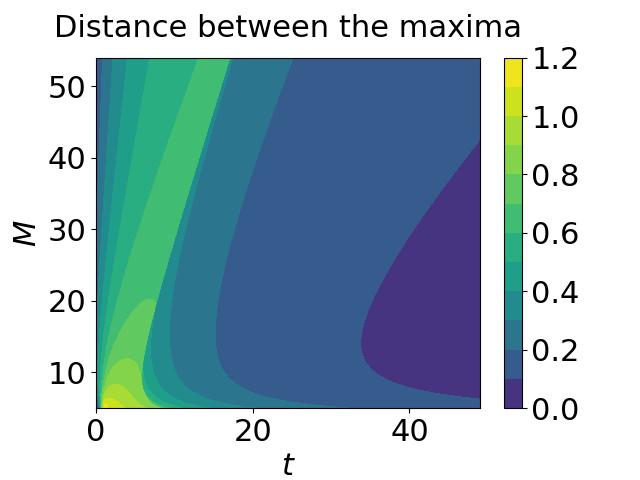}
    \hspace{0cm}
        \centering
        \includegraphics[width=0.30\textwidth]{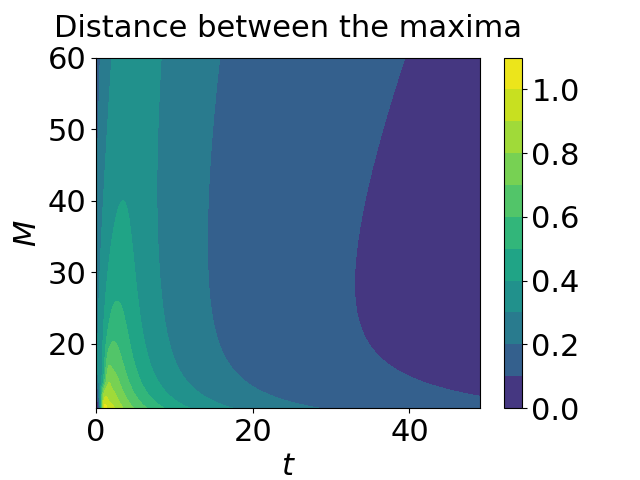}
    \caption{MC simulations comparing the generalization capabilities of the unsupervised regularized model with properties of the spectrum, for ranks $\alpha M=2$ (left) and $\alpha M=5$ (right). In the upper row, generalization experiments are conducted: the color maps report the value of the archetype overlap ($m_{\bm \xi} = \frac1N \sum_{i=1}^N \sigma^{(\infty)}_i\xi^\mu_i$) computed upon convergence to fixed points; these were run similarly to those shown in Figs. \ref{fig:gen_rm} and \ref{fig:gen_rm_manyts}, but for the points defined by equation \eqref{eq:sample}, and starting from new examples instead of close to stored ones. The final overlaps relative to the corresponding archetypes are then measured and averaged over $50$ independent samples for each point. The black dashed line corresponds to one level set of the distance between the maxima of the peaks in the spectrum ($0.17$ on the left and $0.13$ on the right). The red line corresponds to the maxima in $t$ of the function $\phi$ defined in \eqref{eq:phi} for each value of $M$. On the bottom, the level sets of the distance between the maxima of the peaks are shown. 
    }
    \label{fig:gen_dream_levels}
\end{figure*}

To establish a feasible criterion for determining the optimal regularization $t^*$, we now impose the additional requirement that the maxima of the left peak be on the right of its center of mass: this condition pushes most of the eigenvalues belonging to the noisy bulk closer to $0$, so that they contribute less in the mode expansion. 
Thus, rather than simply maximizing the distance between the bulks, we instead maximize
\begin{equation}\label{eq:phi}
    \phi_{\alpha, r, M}(t)=\left(\lambda_{\text{max},1}-\lambda_{\text{CM},1}\right)\left(\lambda_{\text{max},2}-\lambda_{\text{max},1}\right),
\end{equation}
where the indices $1$ and $2$ represent the left and right peaks, respectively. This effectively rules out values for which $\lambda_{\text{max},1}<\lambda_{\text{CM},1}$, since $\phi$ is negative, while otherwise keeping an interplay between both distances. The maximum of $\phi$ for each $M$ is drawn in red in Fig. \ref{fig:gen_dream_levels}, where we see that it lies well within the generalization region.

\section{Conclusions and outlook}\label{sec:conclusions}
%
%
Unlike the supervised setting, where the limiting free-energy density can be computed using standard tools from statistical mechanics \cite{AAKBA-EPL2023}, the unsupervised case remains considerably more challenging. In fact, intrinsic correlations among the stored patterns hinder a comprehensive analytical treatment and make it difficult to theoretically predict whether the network can successfully retrieve the ground-truth archetypes underlying a noisy dataset. 
On the other hand, spectral methods have long played a central role in statistical inference, spike-detection, and neural-network theory, where eigenvalue distributions often reveal the separation between informative directions and noise subspaces \cite{BBP,BenaychGeorges2011,Johnstone2001,BBP-Vivo,Valigi_2025,Ribeiro-2025}. Analogous perspectives have recently proved fruitful also in associative memories, where the algebraic properties of the coupling matrix can anticipate retrieval performance and the emergence of spurious states, see e.g., \cite{kanaka-PRL2006,AABF-JPA2019,Haiping-PRE2021,martin2021implicit,agliari2024spectral,ADF2026_BeingEmpty,benedetti2023eigenvector}.\par
In this work, we show how spectral methods may be used to characterize the behavior of the unsupervised Hopfield model. First, by computing the limiting spectral density of the coupling matrix, we find that it displays a clear transition: for sufficiently informative, low-entropy datasets, the density function attains a double-bulk structure, where one spectral component captures the archetypal directions of the data, while the other accounts for noise. This separation effectively captures the possibility of retrieving the underlying dataset and thus functions as a necessary condition for generalization. In other words, when this separation disappears, retrieval of the underlying patterns is lost.
We also showed, via numerical simulations, that this condition can possibly be made sufficient by tuning an appropriate regularization parameter. We then provided a heuristic formula for the optimal choice of this parameter, ensuring a high-quality generalization, thus making our spectral criteria an accurate predictor of the retrieval capabilities of the network.\par
Beyond this specific model, our results suggest that learning, memory, and generalization can be diagnosed geometrically through spectra. Extending these ideas to benchmark datasets and dense associative architectures offers a promising direction for future work.

\section{Acknowledgements}
E.A. and A.F. acknowledge support from PNRR MUR project PE0000013-FAIR.
P.D.M. acknowledges financial support from the PNRR MUR Project B53C23002010006.
P.V. acknowledges support from UKRI FLF Scheme (No. MR/X023028/1).


\newpage
\clearpage
\onecolumn

\begin{center}
{\Large \bfseries Supplementary Material}\\[1em]
{\large \maintitle}
\end{center}

\vspace{1cm}


\section{Derivation of the asymptotic spectral measure of the unsupervised Hebbian ensemble}
Our goal here is to compute the asymptotic spectrum of the interaction matrix of the unsupervised Hopfield model given in \eqref{eq:J_D_r}, with the statistics of the dataset given in equations \cref{eq:stat_xi,eq:def_eta}. To achieve this, we resort to the Edwards-Jones formula \cite{VivoRMT,Vivo_replicas,pehlevan_wishartreplicas}, which gives the spectrum of an $N\times N$ random matrix $\bm X$ as
\begin{equation}\label{eq:EJ}
    \rho_{N}(\lambda)=-\frac2{\pi N}\lim_{\epsilon\to0^+}\Im\frac\partial{\partial\lambda_\epsilon}\left<\log Z_N(\lambda_\epsilon)\right>_{\bm X},
\end{equation}
with $ \lambda_\epsilon\doteq\lambda-i\epsilon$ and
\begin{equation}
    Z_N(\lambda_\epsilon)\doteq\int_{\mathbb{R}^N}d\bm{y}\exp\left[-\frac{i}{2}\bm{y}^T\left(\lambda_\epsilon\mathds{1}-\bm X\right)\bm{y}\right].
\end{equation}
Since we are interested in the properties of the model in the thermodynamic limit $N\to\infty$, we shall focus on the limiting spectral density
\begin{equation}
    \rho(\lambda)\doteq\lim_{N\to\infty}\rho_N(\lambda).
\end{equation}
To tackle \eqref{eq:EJ}, we employ the replica trick \cite{Vivo_replicas,pehlevan_wishartreplicas}, where we first compute $\left<Z_N^n\right>$ for $n\in\mathbb N$ and then use
\begin{equation}\label{eq:RT}
    \left<\log Z_N\right>_{\bm X}=\lim_{n\to0}\frac{\log\left<Z_N^n\right>_{\bm X}}{n},
\end{equation}
assuming that analytic continuation holds.\footnote{In practice this is not enough; we will also need to assume that the limits $n\to0$ and $N\to\infty$ commute. This is far from obvious, but standard in replica computations and, as we shall see, the results perfectly match experimental evidence.} In our setting (using the notation adopted in the main text), $\bm X\equiv \tilde{\bm J}^{H'}$ and $\rho \equiv \tilde \rho'$, the reason being the replica computations are slightly more convenient keeping the diagonal in the coupling matrix defined in Eq. \eqref{eq:J_D_r}. The average $\langle \cdot \rangle_{\bm X}$ thus corresponds to the expectation w.r.t. the realization of the examples according to the definition in Eq. \eqref{eq:def_eta}. Since
\begin{equation}\label{eq:J_unsup_d}
    {\tilde{\bm J}^{H'}}= \tilde {\bm J}^H+\frac KN\bm I,
\end{equation}
removing diagonal contributions simply shifts the spectral density by $\alpha$, namely $\tilde\rho(\lambda)=\tilde\rho'(\lambda+\alpha)$. Proceeding now with the computation, we have
\begin{align*}
\left<Z_N^n\right>_{\tilde{\bm \xi}}&=\langle
\int_{\mathbb{R}^{Nn}}\big(\prod_{s=1}^nd\bm{y}_s\big)\exp\Big(-\frac{i}{2}\sum_{s=1}^n\bm{y}_s^T(\lambda\bm{I}-{\tilde{\bm J}^{H'}})\bm{y}_s\Big)\rangle_{\tilde{\bm \xi}}=\\
&=\int d\bm yd\mu_{\bm\chi}d\mu_{\bm\xi}\exp\Big(-\frac{i}{2}\sum_{s=1}^n\sum_{i=1}^N \lambda\left(y_{is}\right)^2\Big)\exp\Big(\frac{i}{2NM}\sum_{s=1}^n\sum_{i,j=1}^N \sum_{\mu,a=1}^{K,M}y_{is}\xi^\mu_i\xi^\mu_j\chi^\mu_{ia}\chi^\mu_{ja}y_{js}\Big).
\end{align*}
To deal with the expectation w.r.t. the disorder $\tilde{\bm \xi}$, we use the generalization of the Hubbard-Stratonovich transform
\begin{equation}\label{eq:HS}
e^{\frac{i}{2} a x^2}
= \frac{e^{i \pi / 4}}{\sqrt{2 \pi }}
\lim_{\zeta\downarrow0}\int_{-\infty}^{\infty}  dz\,
\exp\!\Big(-\frac{i+\zeta}{2 } z^2 + i\,\sqrt{a} z x \Big) .
\end{equation}
Also denoting
\begin{eqnarray}
    \mathcal D\bm y&=&d\bm y\exp\Big(-\frac{i}{2}\sum_{s=1}^n\sum_{i=1}^N \lambda\left(y_{is}\right)^2\Big),\\
     \mathcal{D}\bm z(\zeta)&=&\prod_{s,a,\mu=1}^{n,M,K}\frac{e^{i \pi / 4}}{\sqrt{2 \pi }} dz^\mu_{as}\exp\!\Big( -\frac{i+\zeta}{2 } (z^\mu_{as})^2 \Big)
\end{eqnarray}
we get
\begin{align*}
    \left<Z_N^n\right>_{\tilde{\bm \xi}}&=\int \mathcal D\bm yd\mu_{\bm\chi}d\mu_{\bm\xi}\prod_{s,\mu,a=1}^{n,K,M}\exp\Big(\frac{i}{2NM}\big(\sum_{i=1}^N y_{is}\xi^\mu_i\chi^\mu_{ia}\big)^2\Big)=\\
    &=\int \mathcal D\bm yd\mu_{\bm\chi}d\mu_{\bm\xi}\lim_{\zeta\downarrow0}\prod_{s,\mu,a=1}^{n,K,M}\int\frac{e^{i \pi / 4}}{\sqrt{2 \pi }} dz^\mu_{as}\exp\!\Big( -\frac{i+\zeta}{2 } (z^\mu_{as})^2 + \frac{i}{\sqrt{NM}}  z^\mu_{as}\sum_{i=1}^N y_{is}\xi^\mu_i\chi^\mu_{ia} \Big)=\\
    \nonumber&=
    \lim_{\zeta\downarrow0}\int \mathcal D\bm yd\mu_{\bm\xi}\mathcal{D}\bm z(\zeta)\prod_{i,\mu,a=1}^{N,K,M}\Big(\frac{1+r}{2}\exp\!\big(\frac{i}{\sqrt{NM}}  \sum_{s=1}^n y_{is}z^\mu_{as}\xi^\mu_i \big)+\frac{1-r}{2}\exp\!\big(-\frac{i}{\sqrt{NM}}  \sum_{s=1}^n y_{is}z^\mu_{as}\xi^\mu_i\big)\Big)=\\
    &= \lim_{\zeta\downarrow0}\int \mathcal D\bm yd\mu_{\bm\xi}\mathcal{D}\bm z(\zeta)\prod_{i,\mu,a=1}^{N,K,M}\exp\Big(\log\cos\!\big(\frac{1}{\sqrt{NM}}  \sum_{s=1}^n y_{is}z^\mu_{as}\xi^\mu_i \big)+\log\big(1+ir\tan\!\big(\frac{1}{\sqrt{NM}}  \sum_{s=1}^n y_{is}z^\mu_{as}\xi^\mu_i\big)\big)\Big),
\end{align*}
where we averaged w.r.t. the dataset noise $\bm \chi$. Expanding now the above expression up to second order and perform the $\bm \xi$ disorder, we get 
\small{
\begin{align}
    \nonumber\left<Z_N^n\right>_{\tilde{\bm \xi}}&\approx\lim_{\zeta\downarrow0}\int \mathcal D\bm y\mathcal{D}\bm z(\zeta)\prod_{i,\mu=1}^{N,K}\Big<\exp\Big(-\frac{1-r^2}{2NM}  \sum_{s,\ell=1}^n\sum_{a=1}^M y_{is}y_{i\ell}z^\mu_{as}z^\mu_{a\ell} +\frac{ir}{\sqrt{NM}}  \sum_{s,a=1}^{n,M} y_{is}z^\mu_{as}\xi^\mu_i\Big)\Big>_{\bm\xi}=\\
    \nonumber&=\lim_{\zeta\downarrow0}\int \mathcal D\bm y\mathcal{D}\bm z(\zeta)\prod_{i,\mu=1}^{N,K}\exp\Big(-\frac{1-r^2}{2NM}  \sum_{s,\ell=1}^n\sum_{a=1}^M y_{is}y_{i\ell}z^\mu_{as}z^\mu_{a\ell} +\log\cos\big(\frac{r}{\sqrt{NM}}  \sum_{s,a=1}^{n,M} y_{is}z^\mu_{as}\big)\Big),
\end{align}
}
with $\approx$ meant as equality up to negligible contributions in the thermodynamic limit. 
Since $M$ is fixed, we can again expand the last contribution as
\begin{align}
    \nonumber\log\cos\big(\frac{r}{\sqrt{NM}}  \sum_{s,a=1}^{n,M} y_{is}z^\mu_{as}\big)&=-\frac{r^2}{2NM}\sum_{s,\ell=1}^n\sum_{a,b=1}^My_{is}y_{i\ell}z^\mu_{as}z^{\mu}_{b\ell}+\mathcal{O}(N^{-2}).
\end{align}
Hence, again dropping non-leading contributions yields
\begin{equation}
    \left<Z_N^n\right>_{\tilde{\bm \xi}}\approx\lim_{\zeta\downarrow0}\int \mathcal D\bm y\mathcal{D}\bm z(\zeta)\exp\Big( -\frac{1}{2N} \Tr\,\bm y^T\bm z^T\bm A\bm z\bm y\Big)
\end{equation}
with $\bm A = \frac{1-r^2}{M}\bm I_M +\frac{r^2}{M} \bm 1_M \bm 1_M^T$ -- where $\bm 1_M$ is the vector of ones -- with
eigenvalues
\begin{align}
    \label{eq:mu1}\mu_1&=\frac{1-r^2}{M},\\
    \label{eq:mu2}\mu_2&=r^2+\frac{1-r^2}{M},
\end{align}
with multiplicities $M-1$ and $1$ respectively. We can therefore diagonalize the $\bm A$ matrix with an orthogonal transformation, namely $\bm A=\bm P\bm D\bm P^T$
where $\bm D$ is the diagonal matrix with $D_{aa}=\mu_1$, for $a=1,\ldots,M-1$ and $D_{MM}=\mu_2$. The change of variables $\bm u^\mu=\bm P^T\bm z^\mu$ does not affect the Gaussian measure, then we get
\begin{equation}\label{eq:PFdiag}
    \left<Z_N^n\right>_{\tilde{\bm \xi}}\approx\lim_{\zeta\downarrow0}\int \mathcal D\bm y\mathcal{D}\bm u(\zeta)\exp\Big( -\frac{1}{2N}  \Tr \, \bm y^T\bm u^T\bm D\bm u\bm y\Big).
\end{equation}
Now following a process analogous to \cite{pehlevan_wishartreplicas}, we introduce the order parameter
\begin{equation}\label{eq:op}
    \bm X\doteq \frac iN\bm y^T\bm y,
\end{equation}
via
\begin{equation}\label{eq:deltaID}
    1=\int \frac{d\bm Xd\hat{\bm X}}{(4\pi i/N)^{n(n+1)/2}}\exp\Big(-\frac{N}2\Tr\bm X\hat{\bm X}+\frac{i}2\sum_{s,\ell=1}^n\sum_{i=1}^N\hat{\bm X}^{s\ell}y_{is}y_{i\ell}\Big).
\end{equation}
Plugging it into the partition function \eqref{eq:PFdiag} yields
\begin{align}
    \nonumber\left<Z_N^n\right>_{\tilde{\bm \xi}}&\approx\lim_{\zeta\downarrow0}\int\frac{d\bm Xd\hat{\bm X}\mathcal D\bm y\mathcal{D}\bm u(\zeta)}{(4\pi i/N)^{n(n+1)/2}}\exp\Big(-\frac{N}2\Tr\bm X\hat{\bm X}+\frac{i}2\sum_{s,\ell=1}^n\sum_{i=1}^N\hat{\bm X}^{s\ell}y_{is}y_{i\ell}\Big)\\
    &\times\prod_{\mu=1}^K\exp\Big(\frac{i\mu_1}{2}  \sum_{s,\ell=1}^n\sum_{a=1}^{M-1} X_{s\ell}u^\mu_{as}u^\mu_{a\ell}+\frac{i\mu_2}{2}  \sum_{s,\ell=1}^nX_{s\ell}u^\mu_{M,s}u^\mu_{M,\ell}\Big).
\end{align}
We can now integrate w.r.t. $\bm y$ and $\bm u$. For the first integral, we have
\begin{align}
\int\mathcal{D}\bm y\exp\Big(\frac{i}2\sum_{s,\ell=1}^n\sum_{i=1}^N\hat{\bm X}^{s\ell}y_{is}y_{i\ell}\Big)
    \nonumber=\int d\bm y\exp\Big(-\frac{i}2\sum_{s,\ell=1}^n\sum_{i=1}^N(\lambda\delta_{s\ell}-\hat{\bm X}_{s\ell})y_{is}y_{i\ell}\Big)=i^{-n/2}\operatorname{det}^{-1/2}(\lambda\bm I-\hat{\bm X}),
\end{align}
while
\begin{align*}
&\lim_{\zeta\downarrow0}\int\mathcal{D}u(\zeta)\prod_{\mu=1}^K\exp\Big( -\frac{\mu_1}{2}  \sum_{s,\ell=1}^n\sum_{a=1}^{M-1} X_{s\ell}u^\mu_{as}u^\mu_{a\ell}-\frac{\mu_2}{2}  \sum_{s,\ell=1}^nX_{s\ell}u^\mu_{M,s}u^\mu_{M,\ell}\Big)=\\
    =&\lim_{\zeta\downarrow0}\prod_{a,\mu=1}^{M-1,K}\int\prod_{s=1}^n\frac{e^{i \pi / 4}}{\sqrt{2 \pi }} du^\mu_{as}\exp\!\Big(\frac1{2}  \sum_{s,\ell=1}^n \left((i+\zeta)\delta_{s\ell}-i\mu_1X_{s\ell}\right)u^\mu_{as}u^\mu_{a\ell}\Big)\\
    \times&\lim_{\zeta\downarrow0}\prod_{\mu=1}^K\int\prod_{s=1}^n\frac{e^{i \pi / 4}}{\sqrt{2 \pi }} du^\mu_s\exp\!\Big(\frac1{2}  \sum_{s,\ell=1}^n \left((i+\zeta)\delta_{s\ell}-i\mu_2X_{s\ell}\right)u^\mu_su^\mu_\ell\Big)\\
    =&\det\left(\mathds1_n-\mu_1\bm X\right)^{-(M-1)K/2}\det\left(\mathds1_n-\mu_2\bm X\right)^{-K/2}
\end{align*}
Thus, dropping unessential volume factors, we get
\begin{equation}
    \left<Z_N^n\right>_{\tilde{\bm \xi}}\propto\int d\bm Xd\hat{\bm X}\exp\Big(-\frac{nN}{2}S_n(\bm X,\hat{\bm X};\lambda)+\mathcal{O}(1)\Big),
\end{equation}
with
\begin{align}
    \nonumber nS_n(\bm X,\hat{\bm X};\lambda)\doteq&\Tr\bm X\hat{\bm X}+\log\operatorname{det}(\lambda\bm I-\hat{\bm X})+\alpha(M-1)\log\det(\bm I-\mu_1\bm X)+\alpha\log\det(\bm I-\mu_2\bm X).
\end{align}
To apply saddle-point approximation in the limit $N\to\infty$, we now extremize $S_n$ with respect to the order parameters. Our replica-symmetric ansatz is
\begin{align}
    \label{eq:rs1}\bm X&=q\mathds1_n+c\bm 1_n\bm1_n^\intercal,\\
    \label{eq:rs2}\hat{\bm X}&=\hat{q}\mathds1_n+\hat{c}\bm 1_n\bm1_n^\intercal,
\end{align}
where $\bm1_n$ and $\bm 1_n\bm1_n^\intercal$ respectively denote the column vector of dimension $n$ and the $n\times n$ matrix of ones. We can then compute
\begin{equation}\label{eq:matprod}
    \bm X\hat{\bm X}=q\hat q\mathds1_n+(q\hat c+\hat qc)\bm 1_n\bm1_n^\intercal+nc\hat c\bm 1_n\bm1_n^\intercal,
\end{equation}
and thus, up to first order in $n$,
\begin{equation}\label{eq:matprodtr}
    \frac1n\Tr \bm X\hat{\bm X}\approx q\hat q+q\hat c+\hat q c.
\end{equation}
To compute the other terms, we use the matrix determinant lemma:
\begin{equation}
    \det\left(\bm A+\bm u\bm v^\intercal\right)=\left(1+\bm v^\intercal\bm A^{-1}\bm u\right)\det\bm A.
\end{equation}
Thus, we get, for $k=1,2$,
\begin{align}
    \frac1 n\log\det(\mathds1_n-\mu_k\bm X)&=-\frac {c\mu_k}{1-\mu_kq}+\log\left[1-\mu_kq\right]+\mathcal{O}(n),\\
    \frac1n\log\operatorname{det}(\lambda\mathds1_n-\hat{\bm X})&=-\frac{\hat c}{\lambda-\hat q}+\log[\lambda-\hat q]+\mathcal{O}(n).
\end{align}
Hence $S_n(\bm X,\hat{\bm X};\lambda)= S(q,\hat q, c, \hat c;\lambda)+\mathcal{O}(n)$, with
\begin{align}
    \nonumber S(q,\hat q, c, \hat c;\lambda)&= q\hat q+q\hat c+\hat q c-\frac{\alpha(M-1)\mu_1c}{1-\mu_1q}+\alpha(M-1)\log\left[1-\mu_1q\right]\\
   & -\frac{\alpha\mu_2 c}{1-\mu_2q}+\alpha\log\left[1-\mu_2q\right]-\frac{\hat c}{\lambda-\hat q}+\log[\lambda-\hat q].
\end{align}

The corresponding extrema equations are then
\begin{align}
    \frac{\partial S}{\partial q}=0\iff&\hat q^*+\hat c^*-\frac {\alpha(M-1)\mu_1^2c^*}{(1-\mu_1q^*)^2}-\frac{\alpha(M-1)\mu_1}{1-\mu_1q^*}-\frac {\alpha\mu_2^2 c^*}{(1-\mu_2q^*)^2}-\frac{\alpha\mu_2}{1-\mu_2q^*}=0\label{eq:scq_noise}\\
    \label{eq:scqh_noise}\frac{\partial S}{\partial\hat q}=0\iff& q^*+c^*-\frac{\hat c^*}{(\lambda-\hat q^*)^2}-\frac1{\lambda-\hat q^*}=0\\
    \label{eq:scc_noise}\frac{\partial S}{\partial c}=0\iff&\hat q^*-\frac{\alpha(M-1)\mu_1}{1 - \mu_1q^*}-\frac{\alpha\mu_2}{1-\mu_2q^*}=0\\
    \label{eq:scch_noise}\frac{\partial S}{\partial\hat c}=0\iff& q^*-\frac1{\lambda-\hat q^*}=0
\end{align}
where as
\begin{equation}\label{eq:gen_DS}
    \frac{\partial S}{\partial\lambda}=\frac1{\lambda-\hat q}+\frac{\hat c}{(\lambda - \hat q)^2}.
\end{equation}
However, \eqref{eq:scq_noise} and \eqref{eq:scch_noise} together imply that
\begin{equation}
    c^*=\hat c^*=0,
\end{equation}
and thus \eqref{eq:gen_DS} becomes
\begin{equation}
    \frac{\partial S}{\partial\lambda}=\frac1{\lambda-\hat q}=q^*.
\end{equation}
Furthermore, equation \eqref{eq:scch_noise} can be re-written as
\begin{equation}\label{eq:solqhq_noise}
    \hat q^*=\frac{\lambda q^*-1}{q^*},
\end{equation}
and replacing it into \eqref{eq:scc_noise} gives the cubic equation 
\begin{equation}\label{eq:cubic_app}
    a\left(q^*\right)^3+b\left(q^*\right)^2+cq^*+d=0,
\end{equation}
where we defined\footnote{Note that by sending $\mu_1\to0$ and $\mu_2\to1$ (corresponding to the limit $r\to1$), one gets instead a quadratic equation, which leads to the Marchenko-Pastur distribution of the no-noise Hopfield model \cite{MP,AABF-JPA2019,agliari2024spectral}.}
\begin{align}
    a&=\lambda\mu_1\mu_2,\\
    b&=(\alpha M-1)\mu_1\mu_2-\lambda\left(\mu_1+\mu_2\right),\\
    c&=(1-\alpha(M-1))\mu_1+(1-\alpha)\mu_2+\lambda,\\
    d&=-1.
\end{align}
Defining
\begin{align}
    u&=\frac{2b^3-9abc+27a^2d}{54a^3}\\
    v&=\frac{3ac-b^2}{9a^2}
\end{align}
and the discriminant $D\coloneqq u^2+v^3$, we get
\begin{equation}
    \tilde\rho'(\lambda) =
    \begin{cases}
        \frac{\sqrt{3}}{2\pi}\left(\sqrt[3]{\sqrt{D(\lambda)}+u}+\sqrt[3]{\sqrt{D(\lambda)}-u}\right), &\quad  \text{if } D(\lambda)>0, \\
        0,      & \quad\text{otherwise},
    \end{cases}
\end{equation}
which is the main result in the main text. This spectral density was also studied in \cite{burda_signalnoise}, where the condition \eqref{eq:spectral_gap_crit} for the separation of the peaks was computed.\footnote{To get contact with the notation used in \cite{burda_signalnoise}, note that our two components have relative weights $p_1={(M-1)}/{M}$ and $p_2=1/M$.}

\section{Equivalence of the spectral measures}
In this appendix, we provide a proof for the equivalence of the limiting spectral measures $\tilde \rho_t (\lambda)$ and $\tilde \rho'_t(\lambda+\bar \lambda)$. We start by focusing on the quantity
$$
\Delta(\bar\lambda)= \frac1N\lVert {\tilde{\bm J}}^{D'}-({\tilde{\bm J}}^{D}+\bar\lambda\bm I)\lVert_F^2= \frac1N \lVert \bm D_N - \bar\lambda \bm I\lVert_F^2,
$$
where we denoted with $\bm D_N$ the diagonal of ${\tilde{\bm J}}^{D'}$, i.e. ${\tilde{\bm J}}^{D'}= {\tilde{\bm J}}^{D}+\bm D_N$, and $\lVert \bm A \lVert _F= \sum_{i,j} A_{ij}^2$ is the Frobenius norm. Notice that $\Delta(\bar\lambda)$ can be regarded as the asymptotic minimal deviation of $\bm{D}_N$ from the diagonal behavior. Indeed, calling $c\in \mathbb R$ and
\begin{equation}
    \Delta (c) = \frac1N \lVert \bm D_N - c \bm I\lVert_F^2 = \frac1N \sum_{i=1}^N (\tilde J^{D'}_{ii} -c)^2,
\end{equation}
by the principle of least square errors and convergence of first moment, we have
\begin{equation}
    \bar c_N = \text{argmin}_{c\in \mathbb R}\Delta(c) = \frac1N \sum_{i=1}^N \tilde J^{D'}_{ii} = \frac1N \text{Tr}\tilde{\bm  J}^{D'}\to \bar \lambda = \int \lambda \tilde \rho_t'(\lambda)d\lambda.
\end{equation}
Then $\Delta (\bar c_N)\le\Delta(\bar \lambda)\le\Delta (\bar c_N) +(\bar \lambda-\bar c_N)^2$, so that $\lim_N\Delta (\bar c_N) =\lim_N\Delta(\bar \lambda)$. Thus, we can directly focus on the matrix ${\tilde{\bm J}}^{D}+\bar\lambda\bm I$ as the ansatz for the asymptotic behavior of ${\tilde{\bm J}}^{D'}$.
\par\medskip
Recall that the solution of the matrix ODE $(1+t)\dot{\bm J} =\bm J-\bm J^2 $ with $\bm J(0) = \tilde{\bm J}^H$ can be cast in the form \cite{FAB-NN2019}
\begin{equation}
    \bm J^{D}(t) = ({1+t}) \tilde{\bm J}^H \frac{1}{\bm I +t \tilde{\bm J}^H}.
\end{equation}
This means that the regularized coupling matrix $\tilde{\bm J}^{D'}$ can be expressed in terms of the resolvent matrix of $\tilde{\bm J}^{H'}$, by using the identity $\bm A \bm G_{\bm A}(z) =\bm I +z \bm G_{\bm A}(z) $ holding for all $\bm A$. Indeed, we have

\begin{equation}
    \label{eq:app2_1}
    \tilde{\bm J}^{D'}= \frac{1+t}{t}\Big(\bm I - \frac1t\bm G_{\tilde{\bm J}^{H'}}\Big(-\frac1t\Big)\Big).
\end{equation}
This implies a strict relation between the trace of $\bm J^{D'}$ and the Stieltjes transform of $\bm J^{H'}$:
\begin{equation}
    \frac1N \Tr\tilde{\bm J}^{D'}= \frac{1+t}{t}\Big( 1-\frac1t m_{\tilde{\bm J}^{H'}}(-t^{-1})\Big),
\end{equation}
with $m_{\bm A}(z) = \frac1N \Tr(\bm A-z\bm I)^{-1}.$
Assuming the convergence of the empirical spectral distribution, this means (taking the $N\to\infty$ limit):
\begin{equation}
    \bar \lambda = \frac{1+t}{t}\Big( 1-\int \frac{\tilde \rho'(\lambda)d\lambda}{1+t\lambda}\Big).
\end{equation}
Also, from Eq. \eqref{eq:app2_1} it follows that
\begin{equation}
    \Delta(\bar c_N)=\left(\frac{1+t}{t^2}\right)^2\frac1N \sum_i(m_{\tilde{\bm J}^{H'}}(-t^{-1})- G_{\tilde{\bm J}^{H'},ii}(-t^{-1}))^2.
\end{equation}
\begin{Lem}\label{lem:app_lem1}
    For all $z\in \mathbb C\backslash \mathbb R$, the following inequality holds:
\begin{equation}
    \vert m_{{\tilde{\bm J}}^{D'}}(z)-m_{{\tilde{\bm J}}^{D}+\bar \lambda \bm I}(z)\vert\le \frac{\sqrt{ \Delta(\bar\lambda)}}{(\mathfrak{Im}(z))^2}.
\end{equation}
\end{Lem}
\begin{proof}
    By the resolvent identity
\begin{equation*}
\begin{split}    
    \frac{1}{{\tilde{\bm J}}^{D'}-z\bm I}-\frac{1}{{\tilde{\bm J}}^{D}+\bar \lambda \bm I-z\bm I}&= \frac{1}{{\tilde{\bm J}}^{D'}-z\bm I} [{\tilde{\bm J}}^{D'}-({\tilde{\bm J}}^{D}+\bar \lambda \bm I)]\frac{1}{{\tilde{\bm J}}^{D}+\bar \lambda \bm I-z\bm I}=\\
    &=\frac{1}{{\tilde{\bm J}}^{D'}-z\bm I} [\bm D_N -\bar \lambda \bm I]\frac{1}{{\tilde{\bm J}}^{D}+\bar \lambda \bm I-z\bm I}.
\end{split}
\end{equation*}
Taking the normalized trace:
\begin{equation*}
    \begin{split}
m_{{\tilde{\bm J}}^{D'}}(z)-m_{{\tilde{\bm J}}^{D}+\bar \lambda \bm I}(z)&= \frac1N \Tr \Big(\frac{1}{{\tilde{\bm J}}^{D'}-z\bm I} [\bm D_N -\bar \lambda \bm I]\frac{1}{{\tilde{\bm J}}^{D}+\bar \lambda \bm I-z\bm I}\Big).
    \end{split}
\end{equation*}
 Using $\vert\Tr(\bm A\bm B\bm C)\vert \le \sqrt{\text{rank}(\bm C)}\lVert \bm A \lVert_{op} \lVert \bm B \lVert_{op} \lVert \bm C \lVert_F$, we have
\begin{equation}
    \vert m_{{\tilde{\bm J}}^{D'}}(z)-m_{{\tilde{\bm J}}^{D}+\bar \lambda \bm I}(z)\vert \le \frac1{\sqrt{N}}\Big\lVert \frac{1}{{\tilde{\bm J}}^{D'}-z\bm I} \Big\lVert_{op}\cdot \Big\lVert \frac{1}{{\tilde{\bm J}}^{D}+\bar \lambda \bm I-z\bm I} \Big\lVert_{op}\cdot \lVert \bm D_N -\bar \lambda \bm I\lVert_F,
\end{equation}
as $\text{rank}(\bm D_N -\bar \lambda \bm I)\le N$.
Now, both ${\tilde{\bm J}}^{D'}$ and ${\tilde{\bm J}}^{D}+\bar \lambda \bm I$ are real and symmetric (thus Hermitian), from which it follows that
\begin{equation}
    \Big\lVert \frac{1}{{\tilde{\bm J}}^{D'}-z\bm I} \Big\lVert_{op}\le \frac1{\vert\mathfrak{Im}(z)\vert},
\end{equation}
and similarly for $({\tilde{\bm J}}^{D}+\bar \lambda \bm I-z\bm I)^{-1}$. Using these results, and expressing everything in terms of $\Delta$, we get the thesis.
\end{proof}

%

\begin{Theorem}\label{thm:app_thm1}
Let $g(t) = ({\mathcal Q(t)+t^{-1}})^{-1}$, with $\mathcal Q(t) = \lim_{N\to\infty} \frac{1}{NM}\operatorname{Tr}[(\bm I+ t \tilde {\bm C})^{-1}\bm \Gamma]$ and $\bm \Gamma = \mathbb E  \bm\chi_i\bm\chi_i^T$. Then: 
\begin{equation}
   \max_{i\le N}  \vert G_{\tilde{\bm J}^{H'},ii}(-t^{-1})-  g(t)\vert \xrightarrow{a.s.}0.
\end{equation}
\end{Theorem}

\begin{proof}\leavevmode\par\medskip
{\it Resolvent cavity equations.} Let us call the column vector $\tilde{\bm \xi}_i = ( \tilde\xi_i^{1},\dots,\tilde\xi_i ^{KM})^T$ of length $KM$ -- corresponding to the column in the examples matrix at fixed neuron index $i$ -- and $\tilde {\bm \xi}_{\neg i }$ the $KM\times (N-1)$ matrix obtained from $\tilde{ \bm{\xi}}$ upon removing the $i$-th column. This way, one can express the whole unsupervised Hebbian matrix as

\begin{equation}
    \bm J ^{H'}=\begin{pmatrix}
\tilde J_{ii}^{H'} &\frac1{NM} \tilde{\bm \xi}_i^T  \tilde{\bm \xi}_{\neg i}\\
\frac1{NM}\tilde{\bm \xi}_{\neg i} ^T\tilde{\bm \xi}_i   &\tilde {\bm J}_{(i)}^{H'}
\end{pmatrix},
\end{equation}
with $\tilde {\bm J}_{(i)}^{H'}$ being the $N-1\times N-1$ $i$-th minor matrix of  $\tilde {\bm J}^{H'}$. By straightforward application of the Schur complement formula and using $\tilde J_{ii}^{H'}=\alpha$, we have
\begin{equation}
    G_{\bm J^{H'},ii}(z)= \frac{1}{\alpha-z -\frac{1}{NM }\tilde{\bm \xi}_i^T \bm B_{(i)}(z) \tilde{\bm \xi}_i },
\end{equation}
with $\bm B_{(i)}(z) = \frac1{NM}\tilde{\bm \xi}_{\neg i}\bm G_{(i)}(z) \tilde{\bm \xi}_{\neg i} ^T$ (with dimension $KM\times KM$), and $\bm G_{(i)}(z)$ being the resolvent matrix associated to $\tilde {\bm J}_{(i)}^{H'}$. We focus on the quadratic form
\begin{equation}
    F_N (z) = \frac{1}{NM }\tilde{\bm \xi}_i^T\bm B_{(i)}(z) \tilde{\bm \xi}_i .
\end{equation}
Using the expression of $\bm B_{(i)}$ in terms of the coupling matrix $\tilde {\bm J}_{(i)}^{H'}$ and Woodbury and resolvent identities, it is possible to see that
\begin{equation}
    \bm B_{(i)}(z)=\bm I- (\bm I-\frac1z \tilde {\bm C}_{(i)})^{-1},
\end{equation}
where $\tilde C_{(i)}^{l,l'}= \frac1{NM}\sum_{k\neq i}\tilde \xi^l_k\tilde \xi^{l'}_k$. Then
$-\frac{1}{NM }\tilde{\bm \xi}_i^T \bm B_{(i)}(z) \tilde{\bm \xi}_i = \frac{1}{NM }\tilde{\bm \xi}_i^T (\bm I-\frac1z \tilde {\bm C}_{(i)})^{-1}\bm  \tilde{\bm \xi}_i-\alpha$ (where we used again the fact that $\frac1{NM} \tilde{\bm\xi}_i ^T\tilde{\bm\xi}_i =\alpha$),
then
\begin{equation}
    G_{\bm J^{H'},ii}(z)= \frac{1}{\frac{1}{NM }\tilde{\bm \xi}_i^T (\bm I-\frac1z \tilde {\bm C}_{(i)})^{-1}\tilde{\bm \xi}_i -z }.
\end{equation}
The crucial point in this expression is that the resolvent $(\bm I-\frac1z \tilde {\bm C}_{(i)})^{-1}$ is now independent of $\tilde{\bm \xi}_i$, thus conditioning on $\tilde{\bm \xi}_{\neg i}$ it is a deterministic matrix. We now specialize everything at $z=-1/t$ with $t>0$,\footnote{Notice that, expressing $\tilde{\bm C} = \frac1{NM}\tilde{\bm \xi}_i \tilde{\bm \xi}_i ^T +\tilde{\bm C}_{(i)}$ and using Sherman-Morrison formula to $(\bm I+t\tilde{\bm C}_{(i)})^{-1}$, one can easily recover Eq. \eqref{eq:app2_1} for the diagonal entries of $\tilde{\bm J }^{D'}$.} and focus in particular on the quadratic form
\begin{equation}
    \mathcal Q_{N,i} (t) = \frac{1}{NM}\tilde{\bm \xi}_i ^T (\bm I+ t \tilde {\bm C}_{(i)})^{-1}\tilde {\bm \xi}_i,
\end{equation}
so that $G_{\tilde{\bm J}^{H'},ii}(-t^{-1})= ({\mathcal Q_{N,i} (t)+t^{-1}})^{-1}\doteq F_t({\mathcal Q_{N,i} (t)})$. The denominator is always non-vanishing since $(\bm I+ t \tilde {\bm C}_{(i)})^{-1}$ is positive definite. Now, 
\begin{equation}
    \mathcal Q_{N,i} (t) =\frac1{NM}\sum_{\mu\nu}\sum_{ab} \tilde \xi^{\mu}_{a,i} (\bm I+ t \tilde {\bm C}_{(i)})^{-1}_{(\mu a)(\nu b)}\tilde \xi^{\nu}_{b,i}= \frac1N \sum_{\mu \nu }\xi^\mu_i\Big(\frac1M \sum_{ab} \chi ^{\mu}_{a,i} (\bm I+ t \tilde {\bm C}_{(i)})^{-1}_{(\mu a)(\nu b)} \chi^{\nu}_{b,i}\Big) \xi^\nu_i= \frac1N \bm \xi^T_i \bm K_{(i)} \bm \xi_i.
\end{equation}
\par\medskip
{\it Concentration of quadratic forms.} 
Let us now define the sub-$\sigma$-algebras
$\mathcal F_1 =\{\tilde {\bm \xi} \vert\tilde {\bm \xi}_{\neg i} \text{ fixed}\}$ and $\mathcal F_2 =\{\tilde {\bm \xi} \vert\tilde {\bm \xi}_{\neg i}, \bm \chi_i \text{ fixed}\}$ obtained by conditioning the examples at sites $k \neq i$ and, in the latter, also the multiplicative noise at site $i$. With these definitions, the matrix $ \tilde {\bm C}_{(i)}$ is deterministic, and  $\bm K_{(i)}$ is fixed w.r.t. $\mathcal F_2$.
With these definitions, we have
$$
\mathbb E[\mathcal  Q_{N,i}(t)\vert \mathcal F_2]= \frac1N \sum_{\mu\nu} (\bm K_{(i)})_{\mu\nu}\mathbb E_{\bm \xi_i }\xi^\mu_i\xi^\nu_i =\frac1N\Tr \bm K_{(i)},
$$
while
\begin{equation}
    \mathbb E[\mathcal  Q_{N,i}(t)\vert \mathcal F_1]=\frac1N\mathbb E_{\bm \chi_i }\operatorname{{Tr}} \bm K_{(i)}= \frac 1{NM}\operatorname{{Tr}}[(\bm I + t \tilde{\bm C}_{(i)})^{-1} \bm \Gamma],
\end{equation}
with $\bm \Gamma= \mathbb E_{\bm\chi _i}\bm \chi_i \bm \chi_i ^T$ the second-moment (block diagonal) $KM\times KM$ matrix of the noise at site $i$, given by $\Gamma_{(\mu a),(\nu b)}=\delta_{\mu\nu}\left[\delta_{ab}+r^2\left(1-\delta_{ab}\right)\right]$. Notice that, at finite $N$, $\mathbb E[\mathcal  Q_{N,i}(t)\vert \mathcal F_1]$ is still a function of patterns at sites $k\neq i$. Furthermore, concentration inequalities are expected to hold, and so it is natural to compare $\mathcal Q_{N,i}$ with the $i$-averaged counterpart, namely $\mathbb E[Q_{N,i}\vert \mathcal F_1] $. To do this, we consider the fluctuations $\vert Q_{N,i}-\mathbb E[Q_{N,i}\vert \mathcal F_1]\vert$ in the worst case scenario, and estimate
\begin{equation}
    \mathbb P \Big(\max_{i\le N}\Big\vert \frac1N \bm \xi^T_i \bm K_{(i)} \bm \xi_i-\mathbb E [\frac1N \bm \xi^T_i \bm K_{(i)} \bm \xi_i\vert\mathcal F_1] \Big \vert \ge \epsilon \Big\vert \mathcal F_1\Big).
\end{equation}
By subadditivity, this probability is bounded as $\mathbb P (\max_{i\le N}\vert Q_{N,i}-\mathbb E[Q_{N,i}\vert \mathcal F_1]\vert \ge \epsilon \vert \mathcal F_1)\le \sum_{i=1}^N\mathbb P (\vert Q_{N,i}-\mathbb E[Q_{N,i}\vert \mathcal F_1]\vert \ge \epsilon \vert \mathcal F_1)$, so that we can focus on single events. Now, by triangle inequality
\begin{equation}\label{eq:eq2}
    \begin{split}
        &\mathbb P \Big( \Big\vert \frac1N \bm \xi^T_i \bm K_{(i)} \bm \xi_i-\mathbb E [\frac1N \bm \xi^T_i \bm K_{(i)} \bm \xi_i\vert\mathcal F_1] \Big \vert \ge \epsilon \Big\vert \mathcal F_1\Big)\le \\
        \le &\,\mathbb P \Big( \Big\vert \frac1N \bm \xi^T_i \bm K_{(i)} \bm \xi_i-\mathbb E [\frac1N \bm \xi^T_i \bm K_{(i)} \bm \xi_i\vert\mathcal F_2] \Big \vert \ge \frac\epsilon2 \Big\vert \mathcal F_1\Big)+\mathbb P \Big( \Big\vert\mathbb E [\frac1N \bm \xi^T_i \bm K_{(i)} \bm \xi_i\vert\mathcal F_2] -\mathbb E [\frac1N \bm \xi^T_i \bm K_{(i)} \bm \xi_i\vert\mathcal F_1] \Big \vert \ge \frac\epsilon2 \Big\vert \mathcal F_1\Big).
    \end{split}
\end{equation}
For the first contribution, by tower rule we have
\begin{equation}\label{eq:eq2.1}
    \begin{split}
        \mathbb P \Big( \Big\vert \frac1N \bm \xi^T_i \bm K_{(i)} \bm \xi_i-\mathbb E [\frac1N \bm \xi^T_i \bm K_{(i)} \bm \xi_i\vert\mathcal F_2] \Big \vert \ge \frac\epsilon2 \Big\vert \mathcal F_1\Big)= \mathbb E \Big[\mathbb P \Big( \Big\vert \frac1N \bm \xi^T_i \bm K_{(i)} \bm \xi_i-\mathbb E [\frac1N \bm \xi^T_i \bm K_{(i)} \bm \xi_i\vert\mathcal F_2] \Big \vert \ge \frac\epsilon2 \Big\vert \mathcal F_2\Big)\Big\vert \mathcal F_1\Big] .
    \end{split}
\end{equation}
Now, the argument of the $\mathcal F_1$-expectation can be tackled analytically, since $\bm K_{(i)}$ is a deterministic matrix w.r.t. $\mathcal F_2$, while the pattern $\bm \xi_i$ is a zero-mean subgaussian random vector with independent entries and $\lVert \bm \xi_i \lVert _{\psi_2} = 1$. Then, by Hanson-Wright inequality and Eq. \eqref{eq:eq2.1}, we have
\begin{equation}
    \mathbb P \Big( \Big\vert \frac1N \bm \xi^T_i \bm K_{(i)} \bm \xi_i-\mathbb E [\frac1N \bm \xi^T_i \bm K_{(i)} \bm \xi_i\vert\mathcal F_2] \Big \vert \ge \frac\epsilon2 \Big\vert \mathcal F_1\Big)\le2 \,\mathbb  E \Big[ \exp\Big(- c N \min\Big\{\frac{\epsilon^2}{\frac4N \Tr \bm K_{(i)}^2}, \frac{\epsilon }{2\lVert  \bm K_{(i)} \lVert _{op}}\Big\}\Big)\Big\vert \mathcal F_1\Big],
\end{equation}
for some $c>0$. Since small values of $\frac1N \Tr \bm K_{(i)}^2$ and $\lVert \bm K_{(i)} \lVert_{op}$ produce a stronger concentration, we can upper bound the r.h.s. by upper-bounding the norms. Since the matrix has size $K\times K$, it follows that $\frac1N \Tr\bm K_{(i)}^2 \le \alpha  \lVert \bm K_{(i)}\lVert _{op}^2$. Now, we can put $ \bm K_{(i)}$ in the form $\frac1M \bm X_i^T (\bm I+t \tilde{\bm C}_{(i)})^{-1}\bm X_i$ with $\bm X_i$ being the $MK\times K$ matrix with entries $(\bm X_{i})_{(\mu,a),\nu}= \delta_{\mu\nu} \chi^{\mu}_{a,i},$\footnote{Namely, $\bm X_i$ is the matrix obtained by stacking the noise $\bm\chi_i$ in columns according to their class index $\mu$; more precisely, the first $M$ rows of the first column contain $\{\chi^{1}_{a,i}\}_{a=1}^M$, the second $M$ rows of the second column group $\{\chi^{2}_{a,i}\}_{a=1}^M$, and so on.} such that $\bm X_{i}^T \bm X_i=M \bm I_{K}$. Now, since $\tilde {\bm C}_{(i)}$ is PSD and $t\ge 0$, the eigenvalues of $(\bm I+t \tilde{\bm C}_{(i)})^{-1}$ are at most 1, thus $\lVert (\bm I+t \tilde{\bm C}_{(i)})^{-1}\lVert _{op}\le 1$. This implies that $\lVert \bm K_{(i)}\lVert_{op}=\lVert\frac1M \bm X_i^T (\bm I+t \tilde{\bm C}_{(i)})^{-1}\bm X_i\lVert_{op} \le \frac1M \lVert \bm X_i ^T \bm X_i \lVert_{op}\lVert (\bm I+t \tilde{\bm C}_{(i)})^{-1}\lVert_{op}\le 1 $. Using $\frac1N \Tr \bm K_{(i)}^2\le \alpha$ and $\lVert \bm K_{(i)}\lVert_{op}\le 1$, and since $e^{-1/x}$ is increasing for $x \ge 0$, we immediately have the bound

\begin{equation}
    \mathbb P \Big( \Big\vert \frac1N \bm \xi^T_i \bm K_{(i)} \bm \xi_i-\mathbb E [\frac1N \bm \xi^T_i \bm K_{(i)} \bm \xi_i\vert\mathcal F_2] \Big \vert \ge \frac\epsilon2 \Big\vert \mathcal F_1\Big)\le 2 \,\mathbb \exp\Big(- c N \min\Big\{\frac{\epsilon^2}{4\alpha}, \frac{\epsilon }{2}\Big\}\Big)\doteq  \exp\big(- c N g_1(\epsilon)\big).
\end{equation}
As for the second contribution in Eq. \eqref{eq:eq2}, we start by rewriting it as
\begin{equation}
    \mathbb P \Big( \Big\vert\mathbb E [\frac1N \bm \xi^T_i \bm K_{(i)} \bm \xi_i\vert\mathcal F_2] -\mathbb E [\frac1N \bm \xi^T_i \bm K_{(i)} \bm \xi_i\vert\mathcal F_1] \Big \vert \ge \frac\epsilon2 \Big\vert \mathcal F_1\Big)= \mathbb P \Big( \Big\vert \frac1N  \Tr \bm K_{(i)}  -\mathbb E_{\bm \chi_i}\frac1N \Tr \bm K_{(i)} \Big \vert \ge \frac\epsilon2 \Big\vert \mathcal F_1\Big) .
\end{equation}
Now, $\frac1N \Tr \bm K_{(i)}$ is a quadratic form of the noise $\bm \chi_i$ at site $i$. In particular:
\begin{equation}
    \frac1N \Tr \bm K_{(i)} = \frac1{NM}\sum_{\mu} \sum_{ab} \chi ^{\mu}_{a,i} (\bm I+ t \tilde {\bm C}_{(i)})^{-1}_{(\mu a)(\mu b)} \chi^{\mu}_{b,i}= \frac1{NM}\bm \chi_i ^T \bm H_{(i)}\bm \chi_i ,
\end{equation}
with $(\bm H_{(i)})_{(\mu a)(\nu b)}= \delta_{\mu\nu} (\bm I+ t \tilde {\bm C}_{(i)})^{-1}_{(\mu a)(\nu b)}$. Now, since $\bm \chi_i$ has non-centered entries, it is convenient to write $\bm\chi_i = \bm \eta_i + r \bm 1_{MK}$, so that $\mathbb E \bm \eta_i = 0$. With this representation, we have
\begin{equation}
    \frac1N  \Tr \bm K_{(i)}  -\mathbb E_{\bm \chi_i}\frac1N \Tr \bm K_{(i)}= \frac1{NM}\bm \eta_i ^T \bm H_{(i)}\bm \eta_i-\mathbb E_{\bm \eta}\frac1{NM}\bm \eta_i ^T \bm H_{(i)}\bm \eta_i+\frac{2r}{MN}\bm 1^T \bm H_{(i)}\bm \eta_i.
\end{equation}
Again by triangle inequality, one has
\begin{equation}\label{eq:eq3}
    \begin{split}
        &\mathbb P \Big( \Big\vert \frac1N  \Tr \bm K_{(i)}  -\mathbb E_{\bm \chi_i}\frac1N \Tr \bm K_{(i)} \Big \vert \ge \frac\epsilon2 \Big\vert \mathcal F_1\Big) \le \\
        \le &\,\mathbb P \Big( \Big\vert \frac1{NM}\bm \eta_i ^T \bm H_{(i)}\bm \eta_i-\mathbb E_{\bm \eta}\frac1{NM}\bm \eta_i ^T \bm H_{(i)}\bm \eta_i\Big \vert \ge \frac\epsilon4 \Big\vert \mathcal F_1\Big)+\mathbb P \Big( \Big\vert\frac{2r}{MN}\bm 1^T \bm H_{(i)}\bm \eta_i \Big \vert \ge \frac\epsilon4 \Big\vert \mathcal F_1\Big).
    \end{split}
\end{equation}
Proceeding as before, it is possible to show that $\frac1{NM} \Tr  \bm H_{(i)}^2 \le \alpha $ and $\lVert  \bm H_{(i)}\lVert_{op}\le 1$. The first contribution is therefore again bounded by Hanson-Wright inequality as
\begin{equation}
    \mathbb P \Big( \Big\vert \frac1{NM}\bm \eta_i ^T \bm H_{(i)}\bm \eta_i-\mathbb E_{\bm \eta}\frac1{NM}\bm \eta_i ^T \bm H_{(i)}\bm \eta_i\Big \vert \ge \frac\epsilon4 \Big\vert \mathcal F_1\Big)\le 2\exp\Big(-c' MN\min \Big\{\frac{\epsilon^2}{16 \alpha},\frac{\epsilon}{4}\Big\}\Big)\doteq  \exp\big(-  N Mg_2(\epsilon)\big),
\end{equation}
for some $c'>0$. For the second contribution, we use the fact that $\bm 1^T \bm H_{(i)}\bm \eta_i $ is a weighted linear combination of zero-mean i.i.d. bounded random variables (since $\vert \eta_i\vert \le2$), with the weights being numbers in $\mathcal F_1$. Further, the sum of the squares of the coefficients is $\lVert\bm 1 ^T\bm H_{(i)}\lVert^2\le\lVert\bm 1\lVert^2 \lVert\bm H_{(i)}\lVert_{op}^2=KM $. By Hoeffding inequality, it follows that
\begin{equation}
    \mathbb P \Big( \Big\vert\frac{2r}{MN}\bm 1^T \bm H_{(i)}\bm \eta_i \Big \vert \ge \frac\epsilon4 \Big\vert \mathcal F_1\Big) \le \exp \Big(-c'' \frac{M N\epsilon^2}{\alpha r^2}\Big)\doteq  \exp\big(- N Mg_3(\epsilon)\big),
\end{equation}
for some $c''>0$.
Putting all pieces together, one thus has
\begin{equation}
    \mathbb P \Big(\max_{i\le N}\Big\vert \frac1N \bm \xi^T_i \bm K_{(i)} \bm \xi_i-\mathbb E [\frac1N \bm \xi^T_i \bm K_{(i)} \bm \xi_i\vert\mathcal F_1] \Big \vert \ge \epsilon \Big\vert \mathcal F_1\Big)\le N \Big(2 \exp(-N g_1(\epsilon))+2\exp(-NM g_2 (\epsilon))+ \exp\big(-  N Mg_3(\epsilon)\big)\Big).
\end{equation}
Then, $\max_{i\le N}\Big\vert \frac1N \bm \xi^T_i \bm K_{(i)} \bm \xi_i-\mathbb E [\frac1N \bm \xi^T_i \bm K_{(i)} \bm \xi_i\vert\mathcal F_1]\Big\vert\to 0$ in probability (in $\mathcal F_1$). Further, since the upper bound is summable, convergence is almost sure by Borel-Cantelli:
\begin{equation}
    \max_{i\le N}\Big\vert  \mathcal Q_{N,i}(t)-\frac{1}{NM}\operatorname{Tr} [(\bm I+ t \tilde {\bm C}_{(i)})^{-1}\bm \Gamma]\Big\vert \xrightarrow{\mathcal F_1 -a.s.}0.
\end{equation}
Since in the limit $N\to\infty$ removing a single spin from the correlation matrix is irrelevant, we can safely drop the $i$-dependence from the correlation matrix (as it a 1-rank perturbation of order $\mathcal O (N^{-1})$). Then, defining the function $\mathcal Q(t)= \lim_{N\to\infty}\frac{1}{NM}\operatorname{Tr} [(\bm I+ t \tilde {\bm C})^{-1}\bm \Gamma] $, we have
\begin{equation}
    \mathbb P \Big(\lim_{N\to\infty}\max_{i\le N}\vert  \mathcal Q_{N,i}(t)-\mathcal Q(t)\vert =0\Big)=\mathbb E\Big[\mathbb P\Big(\lim_{N\to\infty}\max_{i\le N}\vert  \mathcal Q_{N,i}(t)-\mathcal Q(t)\vert =0\Big\vert \mathcal F_1\Big)\Big]=1,
\end{equation}
thus $\mathcal F_1$-a.s. convergence towards $\mathcal Q$ is indeed a.s. convergence: $\max_{i\le N}\vert  \mathcal Q_{N,i}(t)-\mathcal Q(t)\vert \xrightarrow{a.s.}0.$
\par\medskip
{\it Concentration of resolvent diagonal via Lipschitz.} 
The function $F_t(x) = (x+1/t)^{-1}$ is clearly Lipschitz for $x\ge0$ with Lipschitz constant $L=t^2$,\footnote{The restriction to $x\ge 0$ is not a problem since $(\bm I+ t \tilde {\bm C})^{-1} $ and $\bm \Gamma$ are PSD, thus the trace of the product is non-negative.} since 
\begin{equation}
    \vert \partial_x F_t(x)\vert = \frac{1}{(x+1/t)^2}\le t^2.
\end{equation}
Calling $g(t)= F_t(\mathcal Q(t))$, we have
\begin{equation}
   \max_{i\le N} \vert G_{\tilde{\bm J}^{H'},ii} (-t^{-1})- g(t)\vert=  \max_{i\le N}\vert F_t(\mathcal Q_{N,i})-F_t(\mathcal  Q(t))\vert\le t^2  \max_{i\le N}\vert \mathcal Q_{N,i} -\mathcal  Q(t)\vert \xrightarrow{a.s.}0,
\end{equation}
which proves our claim.
\end{proof}
As a consequence, for almost every realization of the patterns the diagonal entries of the resolvent of $\tilde{\bm J}^{H'}$ (and by Eq. \eqref{eq:app2_1} of $\tilde{\bm J}^{D'}$) converge to the same value. At this stage, the $g(t)$ function may still depend on the realization of the disorder. However, since $m_{\tilde{\bm J}^{H'}}(-t^{-1})= \frac1N \sum_{i=1}^N G_{\tilde{\bm J}^{H'},ii}(-t^{-1})$, we have
\begin{equation}
    \vert m_{\tilde{\bm J}^{H'}}(-t^{-1}) -g(t) \vert= \Big\vert\frac1N \sum_{i=1}^N G_{\tilde{\bm J}^{H'},ii}-g(t) \Big \vert\le \frac1N \sum_{i=1}^N\vert G_{\tilde{\bm J}^{H'},ii}-g(t)\vert\le \max_{i\le N}\vert G_{\tilde{\bm J}^{H'},ii}-g(t)\vert\xrightarrow{a.s.}0.
\end{equation}
the assumed self-averaging of the Stieltjes transform implies that $g(t)$ coincides almost surely with the deterministic limiting Stieltjes transform evaluated at $z=-1/t$.
As a direct consequence, 
we have
\begin{equation*}
    \begin{split}
   0\le     \Delta (\bar \lambda) &\le \Delta (\bar c_N) +(\bar \lambda-\bar c_N)^2=\left(\frac{1+t}{t^2}\right)^2\frac 1N\sum_i(G_{\bm J^{H'},ii}(-t^{-1})-g(t))^2+(\bar \lambda-\bar c_N)^2\le\\
        &\le \left(\frac{1+t}{t^2}\right)^2\max_{i\le N}[G_{\bm J^{H'},ii}(-t^{-1})-g(t)]^2+(\bar \lambda-\bar c_N)^2=\left(\frac{1+t}{t^2}\right)^2\Big(\max_{i\le N}\vert G_{\bm J^{H'},ii}(-t^{-1})-g(t)\vert \Big)^2+(\bar \lambda-\bar c_N)^2
        \xrightarrow{a.s.} 0,
    \end{split}
\end{equation*}
because of Thm. \ref{thm:app_thm1} and $\bar c_N \to \bar \lambda$. By virtue of Lem. \ref{lem:app_lem1}, this implies that
\begin{equation}
    \lim_{N\to\infty}\vert m_{{\tilde{\bm J}}^{D'}}(z)-m_{{\tilde{\bm J}}^{D}+\bar \lambda \bm I}(z)\vert = 0,
\end{equation}
almost surely, and therefore $\tilde{\bm J}^{D}$ and $\tilde {\bm J}^{D'}-\bar \lambda \bm I$ exhibit the same limiting spectral distribution, namely $\tilde \rho_t (\lambda)$ and $\tilde \rho_t'(\lambda+\bar \lambda)$. As numerical evidence that the diagonal of the interaction matrix $\tilde{ \bm J}^{D'}$ of the unsupervised dreaming model (see eq. \ref{eq:J_D_r}) self-averages in the thermodynamic limit for $t>0$, we compute its standard deviation for systems of different sizes. We did so, using the same control parameters as those used in Figure \ref{fig:spectra}, for several sizes between $N=250$ and $N=2000$, and different dreaming times, as is shown in Figure \ref{fig:diags}. We see that the $\sigma\sqrt{N}$, where $\sigma$ denotes the standard deviation of the diagonal (which basically corresponds to $\Delta(\bar c_N)$), remains approximately equal for each $t$ across all experiments.
\begin{figure*}[h!]
\centering
  \centering
  \includegraphics[width=.40\linewidth]{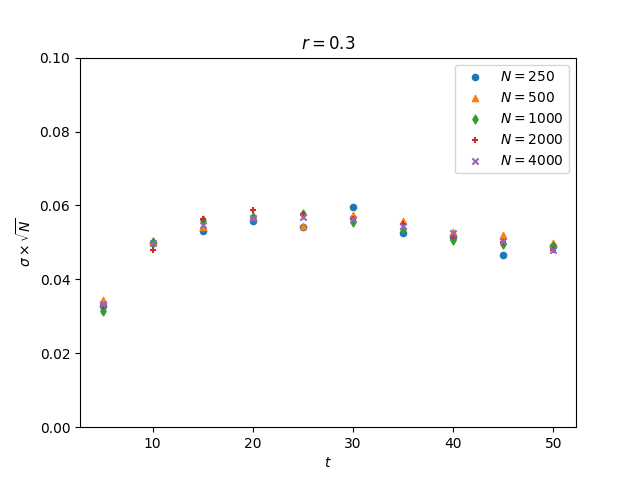}
\hspace{0cm}
  \centering
  \includegraphics[width=.40\linewidth]{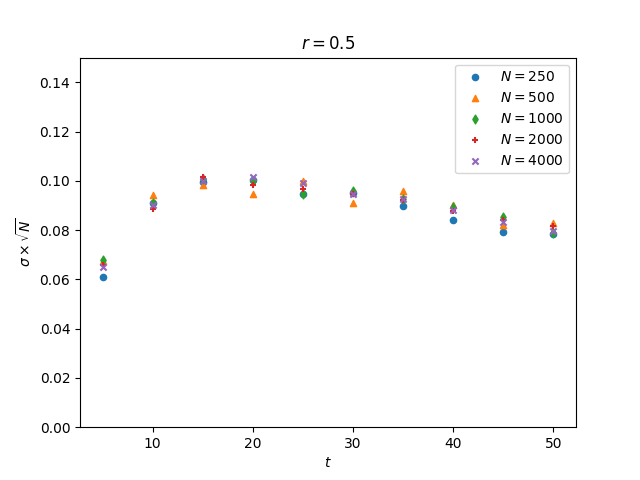}
\caption{Rescaled standard deviation of the diagonal entries of the interaction matrix of the unsupervised regularized Hebbian model for several values of $N$ (varying across the $x$-axis) and $t$ (varying across labels). On the $y$-axis, the standard deviation multiplied by $\sqrt{N}$ is shown. We used $\alpha=0.1$, $M=50$, and $r=0.3$ (left), and $r=0.5$ (right).}
\label{fig:diags}
\end{figure*}

\end{document}